\documentclass[english]{article}
\usepackage[T1]{fontenc}
\usepackage[latin9]{inputenc}
\usepackage{refstyle}
\usepackage{enumitem}
\usepackage{amsmath}
\usepackage{amsthm}
\usepackage{amssymb}
\usepackage{setspace}

\makeatletter

\AtBeginDocument{\providecommand\subsecref[1]{\ref{subsec:#1}}}
\AtBeginDocument{\providecommand\exaref[1]{\ref{exa:#1}}}
\AtBeginDocument{\providecommand\lemref[1]{\ref{lem:#1}}}
\AtBeginDocument{\providecommand\propref[1]{\ref{prop:#1}}}
\AtBeginDocument{\providecommand\secref[1]{\ref{sec:#1}}}
\AtBeginDocument{\providecommand\thmref[1]{\ref{thm:#1}}}
\RS@ifundefined{subsecref}
  {\newref{subsec}{name = \RSsectxt}}
  {}
\RS@ifundefined{thmref}
  {\def\RSthmtxt{theorem~}\newref{thm}{name = \RSthmtxt}}
  {}
\RS@ifundefined{lemref}
  {\def\RSlemtxt{lemma~}\newref{lem}{name = \RSlemtxt}}
  {}

\theoremstyle{plain}
\newtheorem{thm}{\protect\theoremname}[section]
  \theoremstyle{definition}
  \newtheorem{defn}[thm]{\protect\definitionname}
  \theoremstyle{remark}
  \newtheorem{rem}[thm]{\protect\remarkname}
  \theoremstyle{definition}
  \newtheorem{example}[thm]{\protect\examplename}
  \theoremstyle{plain}
  \newtheorem{lem}[thm]{\protect\lemmaname}
  \theoremstyle{plain}
  \newtheorem{prop}[thm]{\protect\propositionname}

\@ifundefined{date}{}{\date{}}
\AtBeginDocument{}
\AtBeginDocument{}
\AtBeginDocument{}
\AtBeginDocument{}
\AtBeginDocument{\providecommand\propref[1]{\ref{prop:#1}}}
\AtBeginDocument{\renewcommand\subsecref[1]{Section~\ref{subsec:#1}}}
\RS@ifundefined{propref}{\newref{prop}{name =Proposition~,names = Propositions~}}{}
\AtBeginDocument{}
\AtBeginDocument{\renewcommand\secref[1]{Section~\ref{sec:#1}}}

\usepackage{tikz}
\usepackage{tocbibind}
\usepackage{youngtab}
\usepackage{authblk}
\usetikzlibrary{matrix}
\usepackage{bbding}

\DeclareMathOperator{\db}{\mathbf{d}}
\DeclareMathOperator{\rb}{\mathbf{r}}

\DeclareMathOperator{\OVL}{OVL}
\DeclareMathOperator{\Bord}{Bord}

\def\RSlemtxt{Lemma~}
\def\RSthmtxt{Theorem~}

\usepackage{babel}

  \providecommand{\definitionname}{Definition}
  \providecommand{\examplename}{Example}
  \providecommand{\lemmaname}{Lemma}
  \providecommand{\propositionname}{Proposition}
  \providecommand{\remarkname}{Remark}
 
\providecommand{\theoremname}{Theorem}

\makeatother

\usepackage{babel}
  \providecommand{\definitionname}{Definition}
  \providecommand{\examplename}{Example}
  \providecommand{\lemmaname}{Lemma}
  \providecommand{\propositionname}{Proposition}
  \providecommand{\remarkname}{Remark}
\providecommand{\theoremname}{Theorem}

\begin{document}

\title{Diamond Subgraphs in the Reduction Graph of a One-Rule String Rewriting
System}

\author{Arthur Adinayev \thanks{Software Engineering Department, Shamoon College of Engineering, Israel}
\\
\Envelope \, arthuad@ac.sce.ac.il\\
 \and Itamar Stein\thanks{Mathematics Unit, Shamoon College of Engineering, Israel (Corresponding
Author) }\\
\Envelope \, Steinita@gmail.com\\
}
\maketitle
\begin{abstract}
In this paper, we study a certain case of a subgraph isomorphism problem.
We consider the Hasse diagram of the lattice $M_{k}$ (the unique
lattice with $k+2$ elements and one anti-chain of length $k$) and
find the maximal $k$ for which it is isomorphic to a subgraph of
the reduction graph of a given one-rule string rewriting system. We obtain
a complete characterization for this problem and show that there is
a dichotomy. There are one-rule string rewriting systems for which
the maximal such $k$ is $2$ and there are cases where there is no
maximum. No other intermediate option is possible.
\end{abstract}


\section{Introduction}

The (directed) reduction graph of a string rewriting system (SRS)
$S$ is the graph whose vertices are words, and whose edges are the
one-step reductions. In this paper we study the reduction graph of
one-rule SRSs. Despite their simple appearance, there are many open
problems regarding one-rule SRSs. For instance, the well-known word
problem asks whether two given words are in the same connected component
of the reduction graph and it is a long standing open problem whether
it is decidable for one-rule SRSs (see \cite[Section 2]{Lallement1995}
for a survey). Another example is the termination problem which asks
if there is an infinite path in the reduction graph and it is also
not known if this problem is decidable for one-rule SRSs (see \cite{Geser2003,Shikishima-Tsuji1997}
and \cite[Problem 21b]{Dershowitz2005}). The reduction graph has
a central role in the treatment of each of these problems and many
other questions regarding SRSs or monoids presented by SRSs. Therefore,
any progress in understanding its structure is of value. 

One way to have a better understanding of a graph $G$ is by finding
basic graphs that are or aren't isomorphic to a subgraph of $G$.
We consider subgraphs to be more related to standard concepts of SRSs
rather than other types of embeddings. For instance, an equivalent
formulation of the termination problem is whether the reduction graph
has a subgraph which is a homomorphic image of the infinite path graph.
In this paper, we take a basic finite graph denoted $M_{k}$ (to be
defined shortly) and consider the question of whether $M_{k}$ is
isomorphic to a subgraph of the reduction graph of a given one-rule
SRS $S$. Except from being an interesting question on its own, this
kind of a problem can lead to new properties and characterizations
of one-rule SRSs that might be of use for other types of questions.
Indeed, some of the properties and notions that appear in this research
(for instance, left\textbackslash{}right cancellativity) have been
used in the study of the word problem and other similar combinatorial
questions (see \cite[Chapter II]{Adjan1966} and \cite{adjan1978}).

The reduction graph of a one-rule SRS $\langle A\mid u\to v\rangle$
is always a graded graph (in the sense that for any two vertices $x,y$
any two paths from $x$ to $y$ has the same length), so clearly it
has only graded subgraphs. The graph we consider in this paper is
the Hasse diagram of the lattice $M_{k}$, where $M_{k}$ is the lattice
with $k+2$ elements $\{x,y,z_{1},\ldots,z_{n}\}$ such that $x\leq z_{i}\leq y$
for $1\leq i\leq n$ and $\{z_{1},\ldots,z_{n}\}$ are pairwise incomparable.
It is clearly a graded graph, and for the sake of simplicity, we denote
it also by $M_{k}$. As already mentioned, given a one-rule SRS $S=\langle A\mid u\to v\rangle$
whose reduction graph is denoted $G_{S}$, our goal is to determine
whether $M_{k}$ is embeddable in $G_{S}$, or in other words, what
is the maximal value of $k$ for which $M_{k}$ is embeddable in $G_{S}$.

The paper is organized as follows. Neglecting few trivial cases ($u=v$
or $|A|=1$) and assuming without loss of generality that $|u|\leq|v|$,
we divide the problem into several cases. In \subsecref{Left-cancellative-SRS}
we prove that if $S$ is left (right) cancellative (i.e., $u$ and
$v$ has different first (respectively, last) letters) then $M_{3}$
is not embeddable in $G_{S}$, hence $k=2$ is maximal. In \subsecref{Unbordered_SRSs}
we generalize this to any system where $v$ is not bordered with $u$,
i.e., $u$ is not a prefix or not a suffix of $v$. In \subsecref{Special_SRSs}
we discuss systems where $u=1$ and prove that if $v\neq b^{n}$ for
every $b\in A$ then $M_{k}$ is embeddable in $G_{S}$ for any natural
$k$. On the other hand, if $v=b^{n}$ for some $b\in A$ then $k=2$
is again the maximum. In \subsecref{Bordered_SRSs} we deal with the
remaining case where $u\neq1$ and $v$ is bordered with $u$. We
use Adyan reduction \cite{adjan1978} to reduce this case to a system
of the form $\langle\tilde{A}\mid1\to\tilde{v}\rangle$ which is the
case solved in \subsecref{Special_SRSs}.

In conclusion, we have obtained a dichotomy between cases where $M_{k}$
is embeddable in $G_{S}$ for every natural $k$ and cases where $k=2$
is the maximal value for which $M_{k}$ is embeddable in $G_{S}$.

\section{Preliminaries}

A \emph{directed} \emph{graph} is a tuple $(V,E,\db,\rb)$ consists
of a set (of vertices) $V$, a set (of edges) $E$ and two functions
$d,r:E\to V$ associating each edge $e\in E$ with a domain vertex
$\db(e)$ and a range vertex $\rb(e)$. A \emph{subgraph} $G^{\prime}=(V^{\prime},E^{\prime},\db^{\prime},\rb^{\prime}$)
of $G$ is a graph such that $V^{\prime}\subseteq V$, $E^{\prime}\subseteq E$
and $\db^{\prime},\rb^{\prime}:E^{\prime}\to V^{\prime}$ are the
corresponding restrictions of $\db$ and $\rb$ (in particular, this
requires that \mbox{$\db(E^{\prime})\subseteq V^{\prime}$} and \mbox{$\rb(E^{\prime})\subseteq V^{\prime}$}).
Let $G_{1}=(V_{1},E_{1},\db_{1},\rb_{1})$ and $G_{2}=(V_{2},E_{2},\db_{2},\rb_{2})$
be two graphs. A \emph{graph homomorphism $f:G_{1}\to G_{2}$ }consists
of two functions $f_{V}:V_{1}\to V_{2}$ and $f_{E}:E_{1}\to E_{2}$
such that 
\[
\db_{2}(f_{E}(e))=f_{V}(\db_{1}(e)),\quad\rb_{2}(f_{E}(e))=f_{V}(\rb_{1}(e))
\]
for every $e\in E_{1}$. We say that $f$ is an embedding (so $G_{1}$
is embedded in $G_{2}$) if $f_{E}$ and $f_{V}$ are injective functions\emph{. }

The set of all words over an alphabet $A$ is denoted by $A^{\ast}$.
We denote the empty word by $1$ and the set of all non-empty words
by $A^{+}$. Let $u,v\in A^{\ast}$ be some words. We say that $u$
is a prefix (suffix) of $v$ if there exists $x\in A^{\ast}$ such
that $v=ux$ (respectively, $v=xu$). Also, $u$ is called a factor
of $v$ if there exist $x,y\in A^{\ast}$ such that $v=xuy$. We say
that $v$ is bordered with $u$ if $u$ is both a prefix and a suffix
of $v$. Recall that the length of a word $u\in A^{\ast}$ is the
number of letters in $u$ and it is denoted $|u|$. We say that the
letter $a\in A$ is \emph{at position $i$ }of $u$ if $u=xay$ for
some $x,y\in A^{\ast}$ and $|x|=i$.

Let $A$ be some set and let $R$ be a relation on $A^{\ast}$. A
tuple $S=\langle A\mid R\rangle$ is called a string rewriting system
(SRS). Elements of $R$ are usually written in the form $u_{i}\to v_{i}$
instead of $(u_{i},v_{i})$. Let $S=\langle A\mid R\rangle$ be an
SRS. The \emph{single-step reduction relation} induced by $R$ is
a relation on $A^{\ast}$ denoted $\to_{R}$ which is defined by $w\to_{R}w^{\prime}$
if $w=xuy$ and $w^{\prime}=xvy$ for some $x,y\in A^{\ast}$ and
$u\to v\in R$. If $|x|=i$ we say that the rule $u\to v$ is being
used at position $i$ in the reduction $w\to_{R}w^{\prime}$. We denote
by $G_{S}$ the \emph{reduction graph} of $S$. It is the (directed)
graph defined as follows. The set of vertices of $G_{S}$ is the set
$A^{\ast}$ of all words over $A$. Given $w,w^{\prime}\in A^{\ast}$,
edges $w\to w^{\prime}$ correspond to tuples $(i,u\to v)$ where
$u\to v$ is a rule in $R$ and $w\to_{R}w^{\prime}$ is a one-step
reduction where $u\to v$ is being used at position $i$. If $S$
has only one rule, namely when $R$ is a singleton, we can identify
an edge only with the position $i$ where the unique rewrite rule
is being used. A path in the reduction graph is called a \emph{reduction
}of $S$.

\section{The embeddability of $M_{k}$ in the reduction graph of a one-rule
SRS\label{sec:Embeddability_section}}
\begin{defn}
Denote by $M_{k}$ the directed graph whose set of vertices is $\{x,y,z_{1},\ldots,z_{k}\}$
and for every $1\leq i\leq k$ there are two edges $x\to z_{i}$ and
$z_{i}\to y$. Note that $M_{k}$ is ``diamond shaped'', for instance,
$M_{3}$ is the Hasse diagram of the diamond lattice:

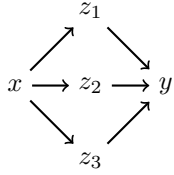
\begin{center}\begin{tikzpicture}

\path (0,0) node (x) {$x$};\path (1,1) node (z_1) {$z_1$};\path (1,0) node (z_2) {$z_2$}; \path (1,-1) node (z_3) {$z_3$};\path (2,0) node (y) {$y$};\draw[thick,->] (x) edge  (z_1); \draw[thick,->] (x)  edge (z_2); \draw[thick,->] (x)  edge (z_3); \draw[thick,->] (z_1) edge  (y); \draw[thick,->] (z_2) edge (y); \draw[thick,->] (z_3) edge   (y);

\end{tikzpicture}\end{center} 
\end{defn}
We want to consider the following question. Given a one-rule SRS $\text{\mbox{\ensuremath{S=\langle A\mid u\to v\rangle}}}$,
what is the maximal $k$ for which $M_{k}$ is isomorphic to a subgraph
of the reduction graph $G_{S}$?

We start with some simple observations. If $u=v$ then the reduction
graph contains only loops and even $M_{1}$ is not a embeddable in
$G_{S}$ so from now on we assume $u\neq v$. If $|A|=1$ then every
connected component of $G_{S}$ with more than one vertex is just
an (infinite) path graph. Therefore only $M_{1}$ is embeddable in
$G_{S}$ and we can assume from now on that $|A|>1$. Another simple
observation is that $M_{k}$ is embeddable in $G_{S}$ for $S=\langle A\mid u\to v\rangle$
if and only if it is embeddable in $G_{S^{-1}}$ where $S^{-1}$ is
the converse system $S^{-1}=\langle A\mid v\to u\rangle$. Therefore,
without loss of generality we can assume that $|u|\leq|v|$.

If an SRS $S=\langle A\mid u\to v\rangle$ satisfies both $|A|>1$ and
$u\neq v$, it is easy to see that $M_{2}$ is embeddable in $G_{S}$.
Indeed, choose a word $w\in A^{\ast}$ such that $uwv\neq vwu$ (for
instance, if $\max\{|u|,|v|\}<l$ we can choose $w=a^{l}b^{l}$).
The reduction graph of $S$ contains the subgraph

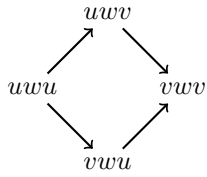
\begin{center}\begin{tikzpicture}

\path (0,0) node (uwu) {$uwu$};\path (1,1) node (uwv) {$uwv$};\path (1,-1) node (vwu) {$vwu$}; \path (2,0) node (vwv) {$vwv$};\draw[thick,->] (uwu) edge  (uwv); \draw[thick,->] (uwu)  edge (vwu);  \draw[thick,->] (uwv) edge  (vwv); \draw[thick,->] (vwu) edge (vwv); 

\end{tikzpicture}\end{center} 

which is isomorphic to $M_{2}$. The question left is whether there
are other values of $k$ for which $M_{k}$ is embeddable in $G_{S}$?
We split this question into several cases.

\subsection{\label{subsec:Left-cancellative-SRS}Left (right) cancellative SRSs}

Let $S=\langle A\mid u\to v\rangle$ be a one-rule SRS such that $u,v\neq1$.
We say that $S$ is \emph{left cancellative} if the first letter of
$u$ and $v$ are different. 
\begin{rem}
The term ``left cancellative'' comes from the well-known fact that
the first letter of $u$ and $v$ are different if and only if the
semigroup presented by $S$ is left cancellative, i.e., $ax=ay$ implies
$x=y$ (see \cite[Chapter II Theorem 2]{Adjan1966}, also stated clearly
in \cite[Theorem 16]{Lallement1995}). 

In this section we will prove that $M_{3}$ is not embeddable in $G_{S}$
if $S$ is a left cancellative SRS.
\end{rem}
Given a reduction of some SRS 
\[
x_{1}\to x_{2}\to\ldots\to x_{n}
\]
we want a way to mark letters that are involved in the rewriting.
For this we introduce a technical tool. Given a set of letters $A=\{a_{1},\ldots,a_{n}\}$
we define a set of ``decorated'' copies $A^{\bullet}=\{a_{1}^{\bullet},\ldots,a_{n}^{\bullet}\}$.
Let $u\in A^{\ast}$ and assume $u=u_{1}\ldots u_{k}$ where every
$u_{i}$ is a letter of $A$. We denote by $u^{\bullet}=u_{1}^{\bullet}\ldots u_{k}^{\bullet}$
a decorated copy of the word $u$. Denote by $\pi:A\cup A^{\bullet}\to A$
a function defined $\pi(a_{i})=\pi(a_{i}^{\bullet})=a_{i}$ which
clearly extends to a projection $\pi:(A\cup A^{\bullet})^{\ast}\to A^{\ast}.$
Now we can define:
\begin{defn}
Let $S=\langle A\mid R\rangle$ be an SRS. Define a new SRS, denoted
$\overline{S}=\langle\overline{A},\overline{R}\rangle$, in the following
way. The set of letters of $\overline{S}$ is $\overline{A}=A\cup A^{\bullet}$.
For every rule $u\to v$ in $R$ and for every word $\overline{u}\in(A\cup A^{\bullet})^{\ast}$
such that $\pi(\overline{u})=u$ the relation $\overline{R}$ will
have the rule $\overline{u}\to v^{\bullet}$.
\end{defn}
\begin{example}
\label{exa:Decoration_Example}If $S=\langle a,b\mid ab\to bba\rangle$
then the SRS $\overline{S}$ is 
\[
\overline{S}=\langle a,a^{\bullet},b,b^{\bullet}\mid ab\to b^{\bullet}b^{\bullet}a^{\bullet},\quad a^{\bullet}b\to b^{\bullet}b^{\bullet}a^{\bullet},\quad ab^{\bullet}\to b^{\bullet}b^{\bullet}a^{\bullet},\quad a^{\bullet}b^{\bullet}\to b^{\bullet}b^{\bullet}a^{\bullet}\rangle
\]
\end{example}
It is obvious that every reduction
\[
\overline{x}_{1}\to\ldots\to\overline{x}_{n}
\]
of $\overline{S}$ can be projected into a reduction of $S$
\[
\pi(\overline{x}_{1})\to\ldots\to\pi(\overline{x}_{n})
\]
by deleting all the ``decorations''. Moreover, it is easy to see
that every reduction of $S$
\[
x_{1}\to x_{2}\to\ldots\to x_{n}
\]
can be ``lifted'' into a reduction of $\overline{S}$
\[
\overline{x}_{1}\to\overline{x}_{2}\ldots\to\overline{x}_{n}
\]
such that $\pi(\overline{x_{i}})=x_{i}$ and $\overline{x_{1}}=x_{1}$.
The decorated letters in this reduction will be the letters that are
``involved'' in the reduction or ``affected'' by it.
\begin{example}
Consider the SRS $S$ in example \exaref{Decoration_Example} and
the reduction
\[
abaabb\overset{(3)}{\to}ababbab\overset{(2)}{\to}abbbabab\overset{(4)}{\to}abbbbbaab
\]
where the numbers over the arrows are the positions in which the rewrite
is being done. This reduction can be lifted to the reduction 
\[
abaabb\overset{(3,ab\to b^{\bullet}b^{\bullet}a^{\bullet})}{\to}abab^{\bullet}b^{\bullet}a^{\bullet}b\overset{(2,ab^{\bullet}\to b^{\bullet}b^{\bullet}a^{\bullet})}{\to}abb^{\bullet}b^{\bullet}a^{\bullet}b^{\bullet}a^{\bullet}b\overset{(4,a^{\bullet}b^{\bullet}\to b^{\bullet}b^{\bullet}a^{\bullet})}{\to}abb^{\bullet}b^{\bullet}b^{\bullet}b^{\bullet}a^{\bullet}a^{\bullet}b
\]

of the SRS $\overline{S}$.

The following observation about reductions in $\overline{S}$ will
be useful.
\end{example}
\begin{lem}
\label{lem:LeftCanLemma}Let $S=\langle A\mid u\to v\rangle$ be a
one-rule SRS and consider a reduction 
\[
x_{1}\to x_{2}\to\ldots\to x_{n}
\]
of $S$ and its lifting 
\[
\overline{x}_{1}\to\overline{x}_{2}\to\ldots\to\overline{x}_{n}
\]
to a reduction of $\overline{S}$. Assume that the first decorated
letter of $\overline{x}_{n}$ is at position $i$ then
\begin{enumerate}
\item \label{enu:NonExistance}No step in the reduction is carried out at
position $j$ for $j<i$.
\item \label{enu:Existance}There is a step in the reduction carried out
at position $i$.
\item \label{enu:Uniqueness}If $S$ is left cancellative then the letter
at position $i$ of $x_{1}$ is the first letter of $u$ and the letter
at position $i$ of $x_{n}$ is the first letter of $v$.
\end{enumerate}
\end{lem}
\begin{proof}
Statements (\ref{enu:NonExistance}) and (\ref{enu:Existance}) are
clear so we will prove (\ref{enu:Uniqueness}). Denote by $a$ the
first letter of $u$ and by $b$ the first letter of $v$. Assume
that in the step $x_{k}\to x_{k+1}$ the rewrite rule is carried out
at position $i$ (such step exists by (\ref{enu:Existance})). Therefore,
the letter at position $i$ of $x_{k}$ is $a$ and the letter at
position $i$ of $x_{k+1}$ is $b$. Since no step is carried out
at position $j$ for $j<i$, the first letter of $x_{1}$ is also
$a$. In addition, the first letter of $u$ and $v$ are different
so we can not carry out any step at position $i$ in the reduction
$x_{k+1}\to\ldots\to x_{n}$. Therefore, the letter at position $i$
of $x_{n}$ is $b$ as required.
\end{proof}
\begin{lem}
\label{lem:SameDecoratedPos}Let $S=\langle A\mid u\to v\rangle$
be a left cancellative SRS and let $x\to z_{1}\to y$ and $x\to z_{2}\to y$
be two reductions in $S$. Denote the corresponding ``lifted'' reductions
in $\bar{S}$ by 
\[
\overline{x}\to\overline{z_{1}}\to\overline{y_{1}},\quad\overline{x}\to\overline{z_{2}}\to\overline{y_{2}}
\]
(a priori, $\overline{y_{1}}\neq\overline{y_{2}}$ because they might
have different decorations). Then, the first decorated positions of
$\overline{y_{1}}$ and $\overline{y_{2}}$ are equal.
\end{lem}
\begin{proof}
Denote by $i_{1}$ ($i_{2}$) the first decorated position of $\overline{y_{1}}$
(respectively, $\overline{y_{2}}$). We continue to use $a$ for the
first letter of $u$ and $b$ for the first letter of $v$. Assume
without loss of generality that $i_{1}<i_{2}$. Applying part (\ref{enu:Uniqueness})
of \lemref{LeftCanLemma} on the reduction $x\to z_{1}\to y$, we
obtain that $b$ is the letter at position $i_{1}$ of $y$ and $a$
is the letter at position $i_{1}$ of $x$. Applying part (\ref{enu:NonExistance})
of \lemref{LeftCanLemma} on $x\to z_{2}\to y$, we obtain that $a$
is the letter at position $i_{1}$ of $y$ (since there are no steps
carried out in this reduction at position $j$ for $j<i_{2}$). This
is a contradiction so $i_{1}=i_{2}$ as required.
\end{proof}
\begin{prop}
\label{prop:M3NotSubGraphOfLeftCanc}Let $S=\langle A\mid u\to v\rangle$
be a left cancellative SRS. Then $M_{3}$ is not isomorphic to a subgraph
of $G_{S}$.
\end{prop}
\begin{proof}
Consider three reductions 
\[
x\to z_{1}\to y,\quad x\to z_{2}\to y,\quad x\to z_{3}\to y
\]
such that $z_{1},z_{2},z_{3}$ are all distinct and lift them into
three reductions in $\overline{S}$
\[
\overline{x}\to\overline{z_{1}}\to\overline{y_{1}},\quad\overline{x}\to\overline{z_{2}}\to\overline{y_{2}},\quad\overline{x}\to\overline{z_{3}}\to\overline{y_{3}}.
\]
According to \lemref{SameDecoratedPos} the first decorated positions
of $\overline{y_{1}}$, $\overline{y_{2}}$ and $\overline{y_{3}}$
are identical. Denote this position by $i$. Part (\ref{enu:Existance})
of \lemref{LeftCanLemma} implies that in each one of the three reduction
there is a rewrite step carried out at position $i$. Without loss
of generality we assume that in the first reduction this is the first
step 
\[
x\overset{(i)}{\to}z_{1}.
\]
In the second reduction this cannot be the first step 
\[
x\overset{(i)}{\to}z_{2}
\]
because this will imply $z_{1}=z_{2}$ in contradiction to our assumption.
Therefore, this must be the second step
\[
z_{2}\overset{(i)}{\to}y.
\]
For the third reduction we cannot have 
\[
x\overset{(i)}{\to}z_{3}
\]
as this implies $z_{1}=z_{3}$ and we cannot have 
\[
z_{3}\overset{(i)}{\to}y
\]
as this implies $z_{2}=z_{3}$. This is a contradiction which finishes
the proof.
\end{proof}
\begin{rem}
Clearly, a dual result holds for right cancellative SRSs.
\end{rem}

\subsection{\label{subsec:Unbordered_SRSs}SRSs where $v$ is not bordered with
$u$}

In this section we generalize the results of \subsecref{Left-cancellative-SRS}
to a wider class of SRSs.
\begin{prop}
Let $S=\langle A\mid u\to v\rangle$ be an SRS such that $u$ is not
a prefix of $v$, then $M_{3}$ is not embeddable in $G_{S}$.
\end{prop}
\begin{proof}
Denote by $p$ the maximal prefix of $u$ which is also a prefix of
$v$. Therefore, we can write $u=pu^{\prime}$ and $v=pv^{\prime}$
for some words $u^{\prime}$, $v^{\prime}$. It might be the case
that $p=1$ (if $S$ is left cancellative) but note that $u^{\prime}\neq1$
since $u$ is not a prefix of $v$ and $v^{\prime}\neq1$ since we
are assuming $|u|\leq|v|$. The maximality of $p$ implies that the
SRS defined by $S^{\prime}=\langle A\mid u^{\prime}\to v^{\prime}\rangle$
is left cancellative. Now, note that any reduction $x\to y$ which
is carried out using the rule $pu^{\prime}\to pv^{\prime}$ can be
carried out using the rule $u^{\prime}\to v^{\prime}$. Therefore,
$G_{S}$ is a subgraph of $G_{S^{\prime}}$. Since $M_{3}$ is not
embeddable in $G_{S^{\prime}}$ by \propref{M3NotSubGraphOfLeftCanc}
it is not embeddable in $G_{S}$ as well.
\end{proof}
Clearly, a dual result holds for SRSs where $u$ is not a suffix of
$v$ so we can conclude:
\begin{prop}
Let $S=\langle A\mid u\to v\rangle$ be an SRS. If $v$ is not bordered
with $u$ (i.e., $u$ is not a prefix of $v$ or not a suffix of $v$)
then $M_{3}$ is not embeddable in $G_{S}$.
\end{prop}

\subsection{\label{subsec:Special_SRSs}Special one-rule SRSs }

In this section we deal with SRSs of the form $S=\langle A\mid1\to v\rangle$.
We remark that SRSs of the form $\langle A\mid v_{i}\to1\rangle$
are called \emph{special} (see \cite[ Definition 3.4.1]{Book1993}).
We have already mentioned that $M_{k}$ is embeddable in $G_{S}$
if and only if it is embeddable in $G_{S^{-1}}$ where $S^{-1}$ is
the converse system. So we can say that in this section we consider
special one-rule SRSs. There are few subcases.
\begin{lem}
\label{lem:v_one_letter}If $v=b^{n}$ for some letter $b\in A$ then
$M_{3}$ is not embeddable in $G_{S}$.
\end{lem}
\begin{proof}
Any word $x\in A^{\ast}$ can be uniquely decomposed into 
\[
x=b^{m_{0}}a_{i_{1}}b^{m_{1}}a_{i_{2}}b^{m_{2}}\cdots b^{m_{l-1}}a_{i_{l}}b^{m_{l}}
\]
where $a_{i_{1}},\ldots,a_{i_{l}}\in A$ are letters distinct from
$b$ and $m_{0},\ldots,m_{l}$ are non-negative integers. If $x\to z$
is a one-step reduction then 
\[
z=b^{m_{0}^{\prime}}a_{i_{1}}b^{m_{1}^{\prime}}a_{i_{2}}b^{m_{2}^{\prime}}\cdots b^{m_{l-1}^{\prime}}a_{i_{l}}b^{m_{l}^{\prime}}
\]
such that $m_{i}^{\prime}=m_{i}+n$ for some $i\in\{0,\ldots,l\}$
and $m_{j}^{\prime}=m_{j}$ if $j\neq i$. It is clear that we can
identify $x$ with the tuple $(m_{0},\ldots,m_{l})$ and a one-step
reduction is equivalent to adding $n$ to one of the entries. Therefore,
a two step reduction $x\to z_{1}\to y$ is equivalent to adding $n$
to two of the entries (or twice to the same one). Now, it is clear
that there could be at most one additional reduction $x\to z_{2}\to y$
from $x$ to $y$ with $z_{1}\neq z_{2}.$ This finishes the proof.
\end{proof}
\begin{lem}
\label{lem:One_To_ab_Case}For any $k\in\mathbb{N}$, the graph $M_{k}$
is embeddable in $G_{S}$ for $\text{\mbox{\ensuremath{S=\langle A\mid1\to ab\rangle}}}$.
\end{lem}
\begin{proof}
Choose $k\in\mathbb{N}$ and take $x=(aabb)^{k-1}$. For $0\leq i\leq k-1$
define $z_{i}=(aabb)^{i}ab(aabb)^{k-i-1}.$ It is clear that $z_{i}$
is obtained from $x$ by applying the rewrite rule at position $4i$.
Moreover, it is clear that $z_{i}\neq z_{j}$ for $i\neq j.$ Now,
applying the rewrite rule at position $4i+1$ we obtain a reduction
$z_{i}\to y$ where $y=(aabb)^{k}$. This yields a subgraph isomorphic
to $M_{k}$ as required.
\end{proof}
\begin{lem}
\label{lem:v_many_letters}Let $S=\langle A\mid1\to v\rangle$ be
an SRS such that $v\neq b^{n}$ for every $b\in A$. Then, $M_{k}$
is embeddable in $G_{S}$ for every $k$.
\end{lem}
\begin{proof}
Assume that the first letter of $v$ is $a$ so $v=av^{\prime}$ where
$v^{\prime}$ contains at least one letter distinct from $a$. Define
a monoid homomorphism $f:\{a,b\}^{\ast}\to A^{\ast}$ which is the
extension of 
\[
f(a)=a,\quad f(b)=v^{\prime}.
\]

It is easy to see that $f$ is injective and that $f(ab)=av^{\prime}=v$.
Therefore, it induces a graph embedding 
\[
\hat{f}:G_{T}\to G_{S}
\]
where $T=\langle a,b\mid1\to ab\rangle$. In particular, it embeds
the subgraph of $G_{T}$ isomorphic to $M_{k}$ (which exists by \lemref{One_To_ab_Case})
onto an isomorphic subgraph of $G_{S}$.

Combining \lemref{v_one_letter} and \lemref{v_many_letters} we conclude
this section.
\end{proof}
\begin{prop}
\label{prop:SpecialSRS}Let $S=\langle A\mid1\to v\rangle$ be an
SRS. If $v=b^{n}$ for some $b\in A$ then $k=2$ is the maximal value
such that $M_{k}$ is embeddable in $G_{S}$. Otherwise, $M_{k}$
is embeddable in $G_{S}$ for every natural $k$.
\end{prop}

\subsection{\label{subsec:Bordered_SRSs}SRSs where $v$ is bordered with $u$}

In this section we will show that any system $S=\langle A\mid u\to v\rangle$
where $v$ is bordered with $u$ can be reduced using Adyan reduction
\cite{adjan1978} into an SRS of the form $\tilde{S}=\langle\tilde{A}\mid1\to\tilde{v}\rangle$
such that $M_{k}$ is embeddable in $G_{S}$ if and only if it is
a embeddable in $G_{\tilde{S}}$. Therefore, we can use \propref{SpecialSRS}
in order to determine whether $M_{k}$ is a subgraph of $S$. We remark
that a similar approach of using Adyan reductions for other one-rule
problems was used in \cite{Shikishima-Tsuji1997} and \cite[Section 6]{Stein2015}. 

We start with some basic definitions required for the reduction.
\begin{defn}
Let $u\in A^{\ast}$ be some word. Its set of \emph{self-overlaps}
is defined by 
\[
\OVL(u)=\{w\in A^{+}\mid\exists x,y\in A^{+}\quad u=xw=wy\}.
\]
The word $u$ is called \emph{self-overlap-free if $\OVL(u)=\varnothing$.}

Let $T$ be a self-overlap-free word over some alphabet $A$. Enumerate
all words in $A^{\ast}$ without $T$ as a factor by 
\[
R_{1},R_{2},\ldots
\]
and let $B$ be an infinite set of new letters
\[
B=\{b_{1},b_{2},\ldots\}\quad(B\cap A=\varnothing).
\]

Denote the set of words bordered with $T$ by $\Bord_{T}$ and note
that every word $x\in\Bord_{T}$ can be decomposed uniquely into 
\[
x=TR_{i_{1}}TR_{i_{2}}\cdots TR_{i_{k}}T.
\]
\end{defn}
Adyan and Oganesyan define a bijection $\varphi_{T}:\Bord_{T}\to B^{\ast}$
inductively by
\[
\varphi_{T}(x)=\begin{cases}
1 & x=T\\
\varphi_{T}(x_{1})b_{i} & x=x_{1}R_{i}T,\quad x_{1}\in\Bord_{T}.
\end{cases}
\]

It is important to observe some properties of $\varphi_{T}$.
\begin{lem}
For every $x\in\Bord_{T}$ we have that $|\varphi_{T}(x)|<|x|$.
\end{lem}
\begin{proof}
This can easily be proved by induction since $0=|\varphi_{T}(T)|<|T|$
and $1=|b_{i}|\leq|R_{i}T|$ even if $R_{i}$ is the empty word. 
\end{proof}
\begin{lem}
Let $u,v\in\Bord_{T}$ such that $u$ is a prefix of $v$, then $\varphi_{T}(u)$
is a prefix of $\varphi_{T}(v)$.
\end{lem}
\begin{proof}
It is clear from the definition of $\varphi_{T}$ that 
\[
\varphi_{T}(Tx_{1}Tx_{2}T)=\varphi_{T}(Tx_{1}T)\varphi_{T}(Tx_{2}T).
\]
Therefore, if $u=T\overline{u}T$ and $v=T\overline{u}TwT$ then 
\begin{align*}
\varphi_{T}(v) & =\varphi_{T}(T\overline{u}TwT)=\varphi_{T}(T\overline{u}T)\varphi_{T}(TwT)\\
 & =\varphi_{T}(u)\varphi_{T}(TwT)
\end{align*}
so $\varphi_{T}(u)$ is indeed a prefix of $\varphi_{T}(v)$.
\end{proof}
A dual argument shows that if $u$ is a suffix of $v$ then $\varphi_{T}(u)$
is a suffix of $\varphi_{T}(v)$. Therefore, we obtain:
\begin{lem}
\label{lem:Adyan_Reduction_Bordered}Let $u,v\in\Bord_{T}$ be distinct
words such that $v$ is bordered with $u$ then $\varphi_{T}(v)$
is bordered with $\varphi_{T}(u)$.
\end{lem}
From now on we consider an SRS $S=\langle A\mid u\to v\rangle$ such
that ($u\neq v$ and) $v$ is bordered with $u$. This implies that
$u\in\OVL(v)$. Denote by $T$ the shortest element of $\OVL(u)$
or $T=u$ if $\OVL(u)=\varnothing$. Clearly, $T$ is self-overlap-free
and $T\in\OVL(v)$ so both $u$ and $v$ are bordered with $T$. (A
system $S=\langle A\mid u\to v\rangle$ with this property is called
\emph{reducible} in \cite{adjan1978}). We make some observations
on the existence of a subgraph of $G_{S}$ isomorphic to $M_{k}$.
\begin{lem}
\label{lem:Subgraph_in_bord_T}If $M_{k}$ is embeddable in $G_{S}$
then it is also isomorphic to a subgraph of $G_{S}$ whose vertices
are in $\Bord_{T}$.
\end{lem}
\begin{proof}
Assume 
\[
x\to z\to y
\]
is a reduction in $G_{S}$. Note that any word $x\in A^{\ast}$ which
contains $T$ as a factor can be written uniquely as $x=x^{\prime}\overline{x}x^{\prime\prime}$
where $\overline{x}\in\Bord_{T}$ and $x^{\prime},x^{\prime\prime}$
do not contain $T$ as a factor. Therefore, we can write the above
reduction as 
\[
x^{\prime}\overline{x}x^{\prime\prime}\to z^{\prime}\overline{z}z^{\prime\prime}\to y^{\prime}\overline{y}y^{\prime\prime}.
\]
Since $u$ and $v$ are bordered with $T$, it is clear that 
\[
x^{\prime}=z^{\prime}=y^{\prime},\quad x^{\prime\prime}=z^{\prime\prime}=y^{\prime\prime}
\]
and
\[
\overline{x}\to\overline{z}\to\overline{y}
\]
is also a reduction. Therefore, if we have $k$ different reductions
\[
x\to z_{1}\to y,\ldots,x\to z_{k}\to y
\]
there are $k$ corresponding reductions
\[
\overline{x}\to\overline{z_{1}}\to\overline{y},\ldots,\overline{x}\to\overline{z_{k}}\to\overline{y}
\]
such that $\overline{x},\overline{y},\overline{z_{1}},\ldots\overline{z_{k}}\in\Bord_{T}$.
Since the steps $x\to z_{i}$ and $x\to z_{j}$ are carried out at
different positions for $i\neq j$ we know that $\overline{x}\to\overline{z_{i}}$
and $\overline{x}\to\overline{z_{j}}$ are carried out in different
positions and hence $\overline{z_{i}}\neq\overline{z_{j}}$. Therefore,
we have a subgraph isomorphic to $M_{k}$ such that all the vertices
are bordered with $T$ as required.
\end{proof}
\begin{lem}
\label{lem:Adyan_Reduction}Let $S=\langle A\mid u\to v\rangle$ be
an SRS such that $v$ is bordered with $u$ and let $T$ be defined
as above. Then $M_{k}$ is embeddable in $G_{S}$ if and only if it
is embeddable in $G_{\hat{S}}$ for $\hat{S}=\langle B\mid\varphi_{T}(u)\to\varphi_{T}(v)\rangle$.
\end{lem}
\begin{proof}
Recall that $\varphi_{T}$ is a bijection $\varphi_{T}:\Bord_{T}\to B^{\ast}$.
It is clear that $\varphi_{T}^{-1}$ maps any subgraph of $G_{\hat{S}}$
onto an isomorphic subgraph of $G_{S}$. On the other direction, if
$G_{S}$ has a subgraph isomorphic to $M_{k}$, then by \lemref{Subgraph_in_bord_T}
it has such subgraph whose vertices are elements of $\Bord_{T}$.
Therefore, $\varphi_{T}$ maps it onto a subgraph of $G_{\hat{S}}$
isomorphic to $M_{k}$ as required.
\end{proof}
\begin{lem}
\label{lem:Infinite_To_Finite_Alphabet}Let $B$ be an alphabet (perhaps
infinite) and let $S=\langle B\mid u\to v\rangle$ be an SRS. Let
$B^{\prime}\subseteq B$ be the (finite) set of letters from $B$
that occur in $u$ and $v$ and define $S^{\prime}=\langle B^{\prime}\mid u\to v\rangle$.
Then, $M_{k}$ is embeddable in $G_{S}$ if and only if it is embeddable
in $G_{S^{\prime}}$.
\end{lem}
\begin{proof}
It is clear that $G_{S^{\prime}}$ is a subgraph of $G_{S}$ by inclusion
so any subgraph of $G_{S^{\prime}}$ is a subgraph of $G_{S}$. On
the other direction denote by $\pi$ the standard projection $\pi:B^{\ast}\to(B^{\prime})^{\ast}$
defined by 
\[
\pi(b)=\begin{cases}
b & b\in B^{\prime}\\
1 & b\notin B.^{\prime}
\end{cases}
\]
It is clear that if 
\[
x\to y
\]
is a reduction of $G_{S}$ carried out at position $i$ then 
\[
\pi(x)\to\pi(y)
\]
is also a reduction of $G_{S^{\prime}}.$ Moreover, the letter at
position $i$ of $x$ is a letter of $B^{\prime}$ (it is the first
letter of $u$). Therefore, if 
\[
x\to z_{1}\to y,\ldots,x\to z_{k}\to y
\]
are $k$ reductions in $G_{S}$ such that $z_{i}\neq z_{j}$ for $i\neq j$
then
\[
\pi(x)\to\pi(z_{1})\to\pi(y),\ldots,\pi(x)\to\pi(z_{k})\to\pi(y)
\]
are $k$ reductions in $G_{S^{\prime}}$ such that $\pi(z_{i})\neq\pi(z_{j})$
for $i\neq j$. This finishes the proof.
\end{proof}
We can now state the main result of this section.
\begin{prop}
\label{prop:Full_Reduction_Bordered_Case}Let $S=\langle A\mid u\to v\rangle$
be an SRS such that $v$ is bordered with $u$ then we can construct
another SRS $\tilde{S}=\langle\tilde{A}\mid1\to\tilde{v}\rangle$
such that $M_{k}$ is embeddable in $G_{S}$ if and only if it is
embeddable in $G_{\tilde{S}}$.
\end{prop}
\begin{proof}
Choose $T$ to be the shortest element of $\OVL(u)$ (or $T=u$ if
$\text{\mbox{\ensuremath{\OVL(u)=\varnothing}}}$). Take $B^{\prime}$
to be the set of letters from $B$ that occur in $\varphi_{T}(u)$
and $\varphi_{T}(v)$. Denote 
\[
A_{1}=B^{\prime},\quad u_{1}=\varphi_{T}(u),\quad v_{1}=\varphi_{T}(v)
\]
and $S_{1}=\langle A_{1}\mid u_{1}\to v_{1}\rangle$. By \lemref{Adyan_Reduction}
and \lemref{Infinite_To_Finite_Alphabet}, $M_{k}$ is embeddable
in $G_{S}$ if and only if it is embeddable in $G_{S_{1}}$. There
is no reason to expect that $\text{\mbox{\ensuremath{u_{1}=1}}}$.
However, by \lemref{Adyan_Reduction_Bordered} $v_{1}$ is still bordered
with $u_{1}$ so we choose $T_{1}$ to be the shortest element of
$\OVL(u_{1})$ or $T_{1}=u_{1}$ if $\text{\mbox{\ensuremath{\OVL(u_{1})=\varnothing}}}$.
Now we can continue this process and construct $S_{2}=\langle A_{2}\mid u_{2}\to v_{2}\rangle$
with $u_{2}=\varphi_{T_{1}}(u_{1})$, $v_{2}=\varphi_{T_{1}}(v_{1})$
and so on. Since $|\varphi_{T}(x)|<|x|$ this process must terminate.
It will terminate when $u_{k}=\varphi_{T_{k-1}}(u_{k-1})=1$. Then
we can define $\tilde{A}=A_{k}$ and $\tilde{v}=v_{k}$ and obtain
a system $\tilde{S}=\langle\tilde{A}\mid1\to\tilde{v}\rangle$ which
satisfy the desired result.
\end{proof}
\begin{rem}
Note that $T$ can be easily obtained from $u$ and $v$ and it is
also a routine procedure to calculate $\langle B^{\prime}\mid\varphi_{T}(u)\to\varphi_{T}(v)\rangle$
(note that $B^{\prime}$ is a finite set). Therefore, the process
described in \propref{Full_Reduction_Bordered_Case} can be effectively
computed.
\end{rem}
\propref{Full_Reduction_Bordered_Case} is enough in order to solve
the case of this section. Given an SRS $S=\langle A\mid u\to v\rangle$
such that $v$ is bordered with $u$ we can carry out the procedure
described in \propref{Full_Reduction_Bordered_Case} and obtain an
SRS $\tilde{S}=\langle\tilde{A}\mid1\to\tilde{v}\rangle$ which is
the case dealt with in \propref{SpecialSRS}.

\section{Conclusion}

Combining the results of \secref{Embeddability_section} we obtain
the following theorem which gives a complete answer to the question
of whether $M_{k}$ is embeddable in the reduction graph of a one-rule
SRS.
\begin{thm}
\label{thm:MainThm}Let $S=\langle A\mid u\to v\rangle$ be a one-rule
SRS such that $u\neq v$, $|u|\leq|v|$ and $|A|>1$. Then:
\begin{enumerate}
\item If $v$ is not bordered with $u$ then $k=2$ is the maximal value
such that $M_{k}$ is embeddable in $G_{S}$.
\item If $v$ is bordered with $u$ then we can use Adyan reductions as
described in \propref{Full_Reduction_Bordered_Case} and obtain an
SRS $\tilde{S}=\langle\tilde{A}\mid1\to\tilde{v}\rangle$. In this
case:
\begin{enumerate}
\item If $\tilde{v}=b^{n}$ for some $b\in\tilde{A}$ then $k=2$ is the
maximal value such that $M_{k}$ is embeddable in $G_{S}$.
\item If $\tilde{v}\neq b^{n}$ for every $b\in\tilde{A}$ then $M_{k}$
is embeddable in $G_{S}$ for every $k$.
\end{enumerate}
\end{enumerate}
\end{thm}
\begin{rem}
Since the procedure described in \propref{Full_Reduction_Bordered_Case}
is effective, \thmref{MainThm} implies that the question of whether
$M_{k}$ is embeddable in $G_{S}$ for a given SRS $S=\langle A\mid u\to v\rangle$
is decidable.
\end{rem}
If $v$ is bordered with $u$ we can consider the SRS $\langle A\mid u\to v\rangle$
as equivalent in some sense to an SRS $\langle\tilde{A}\mid1\to\tilde{v}\rangle$
so it can be considered as a very specific case. Therefore, one way
to interpret \thmref{MainThm} is that $M_{3}$ is not embeddable
in the reduction graph of a ``standard'' one rule SRS. This gives
some restriction on the possible structure of the reduction graph
of a ``typical'' case. It is an interesting question whether other
similar restrictions can be found.

\bibliographystyle{plain}
\bibliography{library}

\end{document}